\documentclass[acmtopc]{acmsmall}

\usepackage[ruled]{algorithm2e}
\SetAlFnt{\small}
\LinesNumbered
\DontPrintSemicolon

\acmVolume{0}
\acmNumber{0}
\acmArticle{0}
\acmYear{2018}
\acmMonth{9}

\usepackage{amsmath,amssymb}
\usepackage{graphicx}
\usepackage{xspace}

\newtheorem{thm}{Theorem}

\numberwithin{equation}{section}

\newcommand{\gplus}{\oplus}
\newcommand{\gminus}{\ominus}

\newcommand{\pvec}[1]{\mathbf{#1}}
\newcommand{\vvec}[1]{\vec{\mathbf{#1}}}

\newcommand{\norm}[1]{\lVert #1 \rVert}

\DeclareMathOperator{\Aut}{\mathrm{Aut}}

\newcommand{\cA}{\mathcal{A}}
\newcommand{\cB}{\mathcal{B}}

\newcommand{\cS}{\mathcal{S}}

\newcommand{\sC}{\mathbb{C}}
\newcommand{\sR}{\mathbb{R}}

\newcommand{\dx}{\mathrm{d}x}

\newcommand{\rRK}{\mathrm{RK}}

\newcommand{\ttx}[1]{\texttt{#1}\xspace}
\newcommand{\ttp}[1]{\texttt{p4est\_#1}\xspace}

\newcommand{\pforest}{\ttx{p4est}}
\newcommand{\pforestbuild}{\ttp{build}}
\newcommand{\pforestlocal}{\ttp{search\_local}}
\newcommand{\pforestparti}{\ttp{search\_partition}}
\newcommand{\pforesttrans}{\ttp{transfer}}

\newcommand{\tfont}[1]{\textbf{#1}}
\newcommand{\treturn}{\tfont{return}\xspace}
\newcommand{\tsegeval}{\tfont{segeval}\xspace}

\newcommand{\alglab}[1]{\label{alg:#1}}
\newcommand{\algref}[1]{Algorithm~\ref{alg:#1}}

\newcommand{\eqnlab}[1]{\label{eqn:#1}}
\newcommand{\eqnref}[1]{\eqref{eqn:#1}}

\newcommand{\figlab}[1]{\label{fig:#1}}
\newcommand{\figref}[1]{Figure~\ref{fig:#1}}

\newcommand{\seclab}[1]{\label{sec:#1}}
\newcommand{\secref}[1]{Section~\ref{sec:#1}}
\newcommand{\appref}[1]{Appendix~\ref{sec:#1}}

\newcommand{\todo}[1]{{\color{blue}(#1)}}

\newcommand{\thetitle}{Distributed-Memory Forest-of-Octrees Raycasting}
\begin{document}

\markboth{C.\ Burstedde}{\thetitle}

\title{\thetitle}
\author{%
CARSTEN BURSTEDDE
\affil{Institut f\"ur Numerische Simulation, Universit\"at Bonn, Germany}
}

\begin{abstract}%
We present an MPI-parallel algorithm for the in-situ visualization of
computational data that is built around a distributed linear forest-of-octrees
data structure.
Such octrees are frequently used in element-based numerical simulations; they
store the leaves of the tree that are local in the curent parallel partition.

We proceed in three stages.
First, we prune all elements whose bounding box is not visible by a parallel
top-down traversal, and repartition the remaining ones for load-balancing.
Second, we intersect each element with every ray passing its box to
derive color and opacity values for the ray segment.
To reduce data, we aggregate the segments up the octree in a strictly
distributed fashion in cycles of coarsening and repartition.
Third, we composite all remaining ray segments to a tiled partition of the
image and write it to disk using parallel I/O.

The scalability of the method derives from three concepts:
By exploiting the space filling curve encoding of the octrees and by relying on
recently developed tree algorithms for top-down partition traversal, we are
able to determine sender/receiver pairs without handshaking and/or collective
communication.
Furthermore, by partnering the linear traversal of tree leaves with the group
action of the attenuation/emission ODE along each segment, we avoid
back-to-front sorting of elements throughout.
Lastly, the method is problem adaptive with respect to the refinement and
partition of the elements
and to the accuracy of ODE integration.
%
%
\end{abstract}





\acmformat{Carsten Burstedde, 2018.  \title.}

\begin{bottomstuff}
The author acknowledges travel support by the Hausdorff Center for Mathematics
(HCM) at the Rheinische Friedrich-Wilhelms-Universit\"at Bonn, Germany.
\end{bottomstuff}

\maketitle

\section{Introduction}
\seclab{intro}

\section{Visualization Model}
\seclab{vismodel}

In this section we begin with defining the geometry of the domain, the camera's
position and properties, and introduce the setup for casting rays through the
geometry to the projection plane.
We focus in detail on the mathematical model for propagating light along a ray
and computing attenuation and emission.
Beginning with the well-known ordinary differential equation (ODE), we propose
the notion of a ray segment that represents the exact effect of the medium
along its length.
We develop the group properties of the ODE's solution into a formalism that
allows to combine, cut, and average arbitrary segments.
Whether segments are separate, touching, or overlapping, the formalism
preserves the physics exactly, which allows us to aggregate all segments of the
ray into the final pixel without requiring global back-to-front processing.
We will rely on this property throughout the paper.

\subsection{Domain description}
\seclab{domaindescription}

We assume that the geometry to be visualized is contained in a space-forest,
that is, a collection of space-trees $\lbrace t_k \rbrace$.  Each quad- or
octree $t_k$ defines an adaptive subdivision of the $d$-dimensional reference
hypercube into its leaves (we use the terms leaf and element interchangeably).
Each node in the tree is either a leaf or the root of a (sub)tree,
and we refer to the union of elements as the mesh.

The data to be visualized may either have been computed in this same forest,
or it may originate from a different source that is covered by the forest and
accessible by element.
We also assume that the space-forest is partitioned in parallel:
Each element and the data it covers is available on exactly one process.
Thus, each element has precisely one owner process, which relates to its owned
elements as the local ones.

We create two- or three-dimensional manifolds by allowing for smooth
transformations
\begin{equation}
  \eqnlab{geometrytrafo}
  \pvec j_k: [0, 1]^d \rightarrow D_k \subset \sR^3
\end{equation}
of the reference hypercube into each of the
tree's domains $D_k$ in real space.  Thus, we define a separate transformation
for each tree, which is then inherited by all elements of a tree.  We can also
speak of the reference cube for an individual element, which is simply a
subcube.
All transformations must be compatible in such a way that
$\cup_k D_k$
makes up the complete geometry.  This
approach allows us to define the geometry with relatively few parameters
compared to storing an independent transformation for each element.
Usually it is possible to replicate the geometry parameters globally, while
there are far too many elements to replicate the mesh in parallel.
If needed, there are ways to avoid replicating per-tree data as well
\cite{BursteddeHolke17}.

In our application we define each transformation $\pvec j_k$ as the tensor
product of degree-$R$ Lagrange polynomials $\psi_m(t)$, $m = 0, \ldots, R$.
For numerical stability, we choose Gau\ss-Lobatto nodes $\eta_m \in [0, 1]$ as
the nodes of the basis polynomials, which means that
\begin{equation}
  \eqnlab{nodal}
  \psi_m(\eta_{m'}) = \delta_{mm'}
  .
\end{equation}
The image of the $d$-dimensional unit cube is defined by $(R + 1)^d$ control
points $\pvec z_{m_1, \ldots, m_d}$.  We derive these points per-tree from any
a-priori definition of the geometry (e.g., a CAD system), producing the formula
\begin{equation}
  \eqnlab{tensorgeom}
  \pvec j_k (t_1, \ldots, t_d) =
  \sum_{m_1, \ldots, m_d = 0}^R \pvec z_{k; m_1, \ldots, m_d}
    \psi_{m_1} (t_1) \cdots \psi_{m_d} (t_d)
  .
\end{equation}
This expression is legal in affine space since Lagrange polynomials form a
partition of unity,
\begin{equation}
  \eqnlab{tensorunity}
  \sum_{m_1, \ldots, m_d = 0}^R
    \psi_{m_1} (t_1) \cdots \psi_{m_d} (t_d) =
  \sum_{m_1} \psi_{m_1} (t_1) \cdots \sum_{m_d} \psi_{m_d} (t_d) =
  1
  , \quad
  t_1, \ldots, t_d \in \sR
  ,
\end{equation}
such that translating points commutes with translating the domain.

Gau\ss-Lobatto nodes have the property that
$\eta_0 = 0$, $\eta_N = 1$.  Together with the nodality \eqnref{nodal} this
implies that the shape of a tree's boundary face, edge, or vertex is defined
by the control points on that part of the boundary alone.  The key observation
is that the geometry of a tree's face in three dimensions has the structure of
\eqnref{tensorgeom} reduced to dimension two, using a strict subset of the
control points.

To visualize two-dimensional manifolds, we need to intersect a straight line
with the image of an element's reference square under the transformation $\pvec
j_k$.  This produces a certain number of intersection points, usually none or
zero (but maybe more for non-affine transformations) that we use to update the
color information for that ray.  For three-dimensional manifolds we intersect a
ray with each of the faces of an element and connect the intersection points by
a path through the element.
Then we update the ray's color information by sampling the elements' properties
along the path segments.

We emphasize that for either two or three space dimensions, the intersection
routine is the same, using \eqnref{tensorgeom} with $d = 2$.  We touch on the
mathematical details in \appref{intersection}.  We also note that it would be
possible to use a different type of geometry transformation altogether, as long
as the intersecton of a two-dimensional mapped square with a ray can be
computed analytically or numerically.
%

When we have a non-affine geometry transformation, an element may produce
multiple path segments that intermix with those from neighboring elements.
Beginning with \secref{raysegments} below, we propose a way of supporting this
without introducing dependencies between any two elements.

\subsection{Camera setup}
\seclab{camerasetup}

\begin{figure}
  \begin{center}
    \includegraphics[height=.5\linewidth]{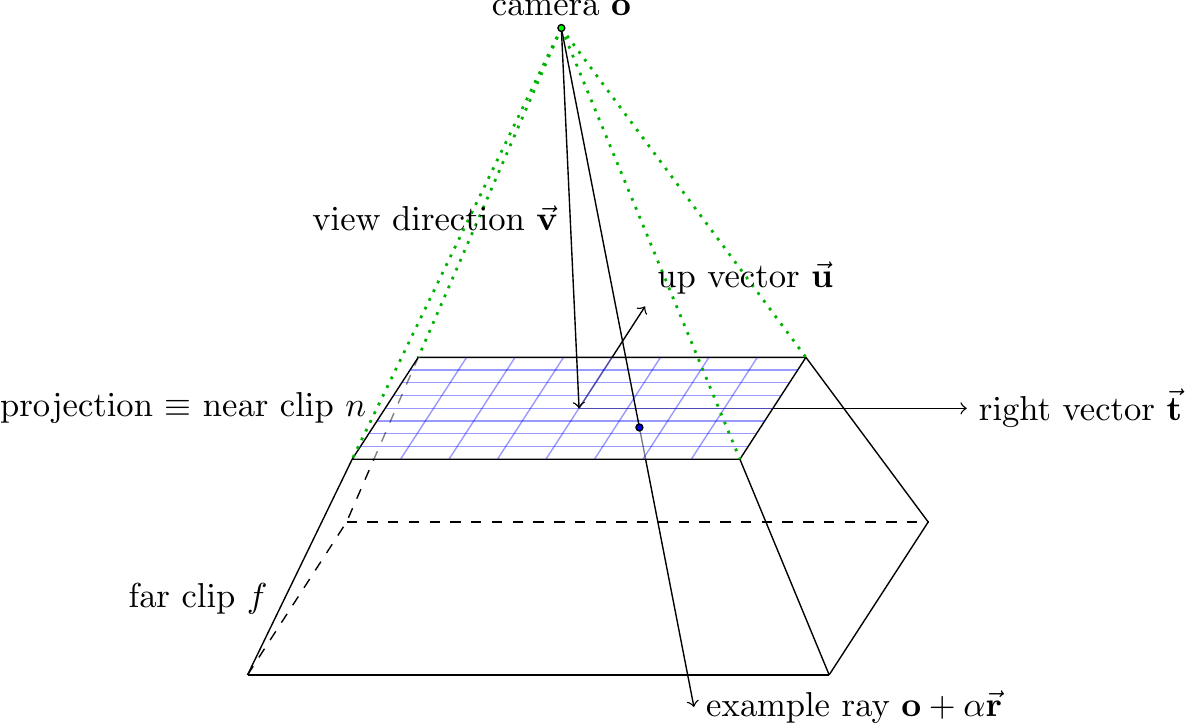}
  \end{center}
  \caption{Our conventions for the camera setup.
     The projection rectangle is divided into pixels (blue), where each pixel
     is associated with a ray through its center.
     The view frustum lies between the near and far clipping planes (at
     distances $n$ and $f$ from the camera) and the sides of the pyramid stump.
     The (left-handed) camera coordinate system is defined by the unit right
     and up-vectors and the normalized view direction,
     $(\vvec t, \vvec u, \vvec v / n)$,
     with the origin $\pvec o$ at the camera.
     A ray with the pixel coordinates $(i, j)$ has the direction vector
     $\vvec r$ given by equation \eqnref{ray}.
     The ray coordinate is $\alpha \ge 0$, where $\alpha = 0$ identifies the
     camera and $\alpha = 1$ the intersection of the ray with its pixel's
     center.
  }%
  \figlab{camera}%
\end{figure}%
In this paper, we use a standard way of defining the camera (see
\figref{camera}).  The camera is defined by a point location in space $\pvec
o$, a view direction $\vvec v$, and a projection/image rectangle perpendicular
to the view direction.  The rectangle has two principal axes defined by the unit
vectors $\vvec t$ (right) and $\vvec u$ (up) that are thought to originate at
its center $\pvec o + \vvec v$.  Thus, the distance $n$ of the projection
rectangle from the camera is $n = \norm{\vvec v}$.  The rectangle is logically
divided into a $w \times h$ grid of pixels, each having length $\ell$ in the
projection plane.  The pixels are indexed by zero-based integers $(i, j)$.
Each pixel is associated with a ray that connects the camera with its center
point and extends beyond the projection plane.  The ray is thus a
set of points along a line parameterized by $\alpha \ge 0$,
\begin{equation}
  \eqnlab{ray}
  \pvec o + \alpha \vvec r ,
  \qquad\text{where}\qquad
  \vvec r = \vvec v + \ell \left( \left( i - \frac{w - 1}2 \right) \vvec t
                                + \left( j - \frac{h - 1}2 \right) \vvec u
                           \right) .
\end{equation}
The unit vector in the view direction is $\vvec v / n$.
We also define near and far clipping planes, where we identify the near
clipping plane with the projection plane for convenience.  The distances from
the camera are thus the aforementioned $n$ for the near plane and $f > n$ for the
far plane.
The visible volume is contained in the view frustum defined by
the clipping planes and the four lines that connect the camera with each of the
corners of the projection rectangle.

To compute the ray intersections with the geometry, it is convenient to
transform all geometry locations into the camera system, where the view
direction is aligned with the negative $z$ axis.  Formally, this
camera-aligned geometry transformation is defined as
\begin{equation}
  \eqnlab{camerageometry}
  \pvec X = \Phi (\pvec x )
          = (\vvec t, \vvec u, \vvec v / n)^{T}
            \left( \pvec x - \pvec o \right)
  .
\end{equation}
After this affine-orthogonal transformation, the visible $z$ values lie in
$\lbrack n, f \rbrack$.

We can avoid the computational expense that is seemingly associated with this
calculation:  Due to our definition of the geometry as per-tree
transformations, it suffices to transform each of the per-tree control points
into the camera system.  Examining \eqnref{tensorgeom} and using
\eqnref{tensorunity}, we establish that the
camera-aligned geometry transformation is given by
\begin{equation}
  \eqnlab{camerapoints}
  \pvec J_k = \Phi \circ \pvec j_k
  \qquad
  \Leftrightarrow
  \qquad
  \pvec Z_{k; m_1, \ldots, m_d} =
  \Phi ( \pvec z_{k; m_1, \ldots, m_d} )
  .
\end{equation}
This mechanism effectively changes the per-element transformations as desired
and can be performed in a preprocessing stage before executing the actual
visualization.

Below, we optimize the visibility tests for an element by bounding box checks.
Bounding boxes are computed axis-aligned (abbreviated AABB) with respect to the
camera system by subjecting a subtree's reference domain to its tree's
camera-aligned geometry transformation $\pvec J_k$ and then computing the
minima and maxima of its extent in each coordinate direction.
For high order geometry transformations this is non-trivial and may be
accelerated by computing a more lenient approximate criterion instead.
Once we know this 3D brick in camera coordinates, we project it into the image
plane to determine the equivalent pixel rectangle, which is our (2D) AABB.

\subsection{Computing ray segments}
\seclab{raysegments}

In a classic raycasting procedure, each ray is thought to originate at infinity
and followed to the camera location, intersecting a projection plane somewhere
in between.  Along its path, the ray aggregates color information by passing
through the geometry.  Each ray can be divided into a
number of line segments defined by the ray's intersection with the elements of
the geometry, or in other words the leaves of the data tree.
The usual way of rendering the ray is
to initialize it with a background color or transparency
and to begin with the segment farthest away from the camera.
The segment modifies the color values seen, which are then processed with the
second-farthest segment and so on, back to front.  Processing a ray in this
sequence guarantees that opacity and color values are computed correctly, but
enforces a sequential processing of segments.


In this paper we reformulate the rendering of a ray such that we can process
the segments in any sequence.
We define the fundamental operation of
aggregating two neighboring ray segments into a longer, equivalent one.  Once
we have aggregated all segments, we can process the background color through
the last effective segment and achieve an identical result as with the classic
procedure.
All we require is that the aggregation is an associative mathematical
operation, which may
limit the general class of rendering methods but applies to the examples
discussed here.  In a sense, we can work with a local instead of a global view
of the ray.

Our specific ray model takes into account emission and absorption separately
for each color channel.
Suppose a channel's emission coefficient along a line in real space
parameterized by $x \in \sR$ is
$\gamma(x)$ in units of intensity over length, and the absorption
coefficient is $\beta(x)$ in units of one over length.
This leads to the ordinary differential equation (ODE)
\begin{equation}
  \eqnlab{Iode}
  I' (x) = \gamma (x) - \beta (x) I (x)
\end{equation}
with known solution
\begin{equation}
  \eqnlab{knownIAB}
  I (x) = A(x) I(x_1) +  B(x)
  ,
\end{equation}
where $x_1$ is arbitrary and
\begin{equation}
  \eqnlab{knownAB}
  A(x) = e^{\int_{x_1}^x - \beta (x') \dx'}
  , \qquad
  B (x) = A(x) \int_{x_1}^x \frac{\gamma (x')}{A (x')} \dx'
  .
\end{equation}
Writing $I_i = I (x_i)$ and evaluating the solution at $x_2 = x_1 + \Delta x$,
we obtain the affine-linear formula
\begin{equation}
  \eqnlab{IABcombine}
  I_2 = A I_1 + B
\end{equation}
depending on the two numbers $(A, B) = f (\beta, \gamma)$,
\begin{equation}
  \eqnlab{ABgeneral}
  A = A (x_2)
    = e^{- \int_{x_1}^{x_2} \beta (x') \dx'}
  , \qquad
  B = B (x_2)
    = \int_{x_1}^{x_2} \gamma (x) e^{- \int_{x}^{x_2} \beta (x') \dx'} \dx
  .
\end{equation}
The transmission ratio $A$ is dimensionless, while $B$ is the intensity added
by emission.
$A$ is always positive by construction.
The special case of a transparent medium, $\beta (x) = 0$, amounts to $A = 1$
and $B = \int \gamma (x) \dx$.

The formulas simplify for constant $\beta$ and $\gamma$, which yield
\begin{equation}
  \eqnlab{ABconstant}
  A = e^{- \beta \Delta x} ,
  \qquad
  B 
    = \frac \gamma \beta \left( 1 - A \right)
  .
\end{equation}
The mathematical limit for $\beta = 0$ is well defined
due to the relation $1 - \exp ( - \epsilon ) \approx \epsilon$,
which produces $B = \gamma \Delta x$.

We can thus uniquely identify a segment with the tuple $T$ consisting of the
pixel index, 
the in and out coordinates $x_i$, and the characteristic
values $(A, B)$ for each channel.
This information corresponds to the exact physics.

\subsection{Aggregating ray segments}
\seclab{raycombine}

Supposing that we have two adjacent segments characterized by $(A_1, B_1)$ for
the farther one and $(A_2, B_2)$ for the one closer to the camera,
it follows from \eqnref{IABcombine} that they are equivalent to the segment
\begin{equation}
  \eqnlab{aggregate}
  (A, B) = (A_2, B_2) \gplus (A_1, B_1)
  = (A_2 A_1, B_2 + A_2 B_1 )
  .
\end{equation}
\begin{thm}
  The aggregation formula \eqnref{aggregate} is associative.
  Indeed, the segments $(A, B)$ form a non-abelian group $\cS$ with
  identity $(1, 0)$ and inverse
  \begin{equation}
    \eqnlab{identityinverse}
    \gminus (A, B) = (1 / A, -B / A)
    .
  \end{equation}
\end{thm}
\begin{proof}
  It suffices to show that $\cS = ( \sR_+ \times \sR, \oplus)$ is the
  semi-direct product of two groups with a non-trivial map $\theta$.
  We propose
  \begin{equation}
    \cA = ( \sR_+, \cdot ) ,
    \qquad \cB = ( \sR, + ) ,
    \qquad \cS = \cA \ltimes_\theta \cB
  \end{equation}
  with the group homomorphism $\theta : \cA \to \Aut (B)$ defined by
  \begin{equation}
    \theta (A) = (B \mapsto AB)
    .
  \end{equation}
  The formula for the inverse is inherited from this definition.
  Alternatively, the associativity, the identity element, and
  \eqnref{identityinverse} are verified by inserting into \eqnref{aggregate}.
\end{proof}
We can use equation \eqnref{aggregate} any number of times to aggregate
adjacent segments into
equivalent larger ones without the need to proceed back-to-front globally.
The segments do not need to touch if they are separated only by empty space.

The aggregation equation \eqnref{aggregate} is also useful when the emission
and absorption coefficients are modeled as piecewise constant within an
element, and we approximate the element's effect on the ray by aggregating the
sub-segments within each element.
Due to the associativity, we can do this for any element without accessing
other parts of the ray.

For a slight change of perspective, we may ask the question, which intensity
$I$ follows from constant absorption and emission coefficients?
We postulate the equilibrium $I = AI + B$ and insert \eqnref{ABconstant}, which
yields
\begin{equation}
  \eqnlab{BIrelation}
  I = \frac \gamma \beta
  \qquad\text{or equivalently}\qquad
  B = B (A, I) = I (1 - A)
  .
\end{equation}
Thus, we can interface the visualization code to an element-based application
in a modular way, following any of two variants:
\begin{itemize}
  \item
    The visualization code computes $x_i$ etc.\ for a segment by
    intersecting a ray with an element and passes it to the application.
  \item
    For each color channel, the application computes (or approximates) $A$ from
    its internal absorption model using the segment's geometry.
    Now, there are two choices:
    \begin{enumerate}
      \item
        \label{itemvisbyB}
        The application computes (or approximates) $B$ from
        \eqnref{ABgeneral} or \eqnref{ABconstant}.
      \item
        \label{itemvisbyI}
        The application computes $I$ as the desired intensity, e.g., by
        applying a transfer function to the data.
        Then it computes $B$ from $I$ by \eqnref{BIrelation}.
    \end{enumerate}
  \item
    The visualization code accepts $(A, B)$ and stores the tuple for future
    aggregation of this segment with other segments.
\end{itemize}
We recall that the approximation of $A$ by numerical integration is rather
straightforward, while that of $B$ is less so.
The alternative (\ref{itemvisbyI}) thus removes a fairly complex computation
from the application, possibly trading in some physical accuracy.
As mentioned above, we can enhance the accuracy of this procedure by
subsampling an element with multiple pieces of the ray segment.
We will discuss proper numerical procedures that deliver provable error bounds
in \secref{numapproxseg}.

Once all segments of one ray are aggregated into final values indexed by the
color channel, $(A_i, B_i)$, we can compute the visible intensity $I_{v,i}$
from the
background intensity $I_{b,i}$ by
\begin{equation}
  I_{v,i} = A_i I_{b,i} + B_i
  .
\end{equation}
The intensity values $I_{v,i}$ are then converted into pixel RGB values.
We see that writing an image with a transparency value is straightforward
without background and emission, that is, when $I_{b,i} = 1$ and $B_i = 0$.
Within a generic rendering framework, $A_i$ and $B_i$ can be fed into its
routines for alpha blending.

\subsection{Cutting and overlapping segments}
\seclab{cutoverlap}

Since some geometrical inaccuracies are often present in practice, it may occur
that we create two partially overlapping segments for the same ray.
We can resolve this situation by cutting the segments into pieces to isolate
the overlapping section and then averaging the segment values in the
overlapping interval.
The first operation required is the split of one segment into two shorter
equivalent ones of given lengths, say $\Delta x_2 + \Delta x_1 = \Delta x$.
Inserting these values into \eqnref{ABconstant} leads to
\begin{thm}
We split a segment with values $(A, B)$ into two segments $(A_i, B_i)$,
$i = 1, 2$, according to
\begin{equation}
  \eqnlab{segsplit}
  A_i = A^{\Delta x_i / \Delta x}
  , \qquad
  B_i = B \frac{1 - A_i}{1 - A}
  .
\end{equation}
The special case $A = 1$ leads to $A_i = 1$ and $B_i = B \Delta x_i / \Delta x$.
\end{thm}
\begin{proof}
  Above formula is consistent with \eqnref{ABconstant}
  and correct in the limit $\beta \Delta x \rightarrow 0$.
  Aggregating the two segments using \eqnref{aggregate} yields the original
  segment $(A, B)$.
\end{proof}
The second operation, computing the average of two segments over the same
interval, can be thought of alternating between infinitesimally small pieces of
each ray.
This motivates
\begin{thm}
  We average the segments $(A_i, B_i)$, $i = 1, 2$ using the geometric mean for
  $A$ and the arithmetic mean for $B$,
  \begin{equation}
    \eqnlab{segaverage}
    A = \sqrt { A_1 A_2 }
    , \qquad
    B = \frac12 ( B_1 + B_2 )
    .
  \end{equation}
\end{thm}
\begin{proof}
  If we divide the length of each segment by a number $2N$, according to
  \eqnref{ABconstant} this amounts to values
  \begin{equation}
    \bar A_i = A_i^\frac1{2N}
    , \qquad
    \bar B_i = B_i \frac{1 - \bar A_i}{1 - A_i} .
  \end{equation}
  Aggregating these two by using \eqnref{aggregate} and the approximation $1 -
  A^q \approx q (1 - A)$ for $q \ll 1$ and $A^q$ near $1$, we obtain
  \begin{equation}
    \bar A = \sqrt{ A_2 A_1 }^\frac1N
    , \qquad
    \bar B = \frac1{2N} \left( B_2 + A_2^\frac1{2N} B_1 \right)
    .
  \end{equation}
  Aggregating $N$ of these results and taking the limit $N \to \infty$ yields
  the claim.
\end{proof}


\subsection{Additional properties of segments}
\seclab{moregroup}

We can express absorption and emission coefficients in terms of the ray
segment's values by computing $f^{-1} (A, B)$,
where we use that $f$ is injective for constant $\beta (x)$ and $\gamma (x)$,
\begin{equation}
  \eqnlab{betagammafromAB}
  \beta = - \frac{\ln  A}{\Delta x}
  \qquad\text{and}\qquad
  \gamma = \frac{\beta B}{1 - A}
  .
\end{equation}
The formula for $\gamma$ has a well-defined limit for $A \to 1$.
This provides a mechanism to average non-constant $\beta(x)$, $\gamma(x)$ to
constants that are physically equivalent, i.e., yield an identical action on
the intensity:
\begin{enumerate}
  \item Compute $(A, B)$ from the general solution \eqnref{ABgeneral}
    to any desired accuracy.
  \item Compute the physical average
    by \eqnref{betagammafromAB}.
\end{enumerate}
By inspecting the definition of $A(x)$, we see that this leads to the
integral average for $\beta$ and a more complex formula for $\gamma$.

Similarly, the aggregation of two segments leads to the weighted averages
\begin{equation}
  \beta \Delta x = \sum_{i = 1}^2 \beta_i \Delta x_i
  \qquad\text{and}\qquad
  \frac\gamma\beta w = \sum_{i = 1}^2 \frac{\gamma_i}{\beta_i} w_i
\end{equation}
with $w = w_2 + w_1$, specifically
\begin{equation}
  w = 1 - A ,
  \qquad w_2 = 1 - A_2 ,
  \qquad w_1 = A_2 (1 - A_1) .
\end{equation}
Again, the emission coeficients $\gamma_i$ are convex combined if
$\beta_1 = \beta_2$ but not in general.
In addition, this formula is not suitable for the common case $A = 1$.
Thus, it is easier in practice to work in terms of $(A, B)$.

From physical intuition, we can neutralize the effect of a segment by running
through it backwards.
We may also ask how to achieve such neutralization by modifying the absorption
and emission coefficients.
\begin{thm}
  Let us introduce the transformation that reverses the interval,
  \begin{equation}
    \eqnlab{swapxy}
    y (x) = x_1 + x_2 - x
    .
  \end{equation}
  Swapping $x_1$ with $x_2$ in \eqnref{ABgeneral} yields the inverse segment
  \eqnref{identityinverse}.
  This result is reproduced by the identity
  \begin{equation}
    \eqnlab{substbetagamma}
    \gminus(A, B) = f ( -\beta \circ y , -\gamma \circ y ) .
  \end{equation}
\end{thm}
\begin{proof}
  Both the swap $x_i \leftarrow y (x_i)$ and \eqnref{substbetagamma} are
  verified by substituting into \eqnref{ABgeneral}.
\end{proof}
For constant absorption and emission, the formula simplifies to just taking
negative $\beta$ and $\gamma$.
This is consistent with both \eqnref{ABconstant} and \eqnref{betagammafromAB}.
%

Choosing negative coefficients is not only a mathematical possibility, but may
be related to physical processes:
A negative absorption coefficient $\beta$ represents a proportional
amplification of light by the medium (laser physics), and a negative $\gamma$
represents an absolute absorption per length unit (inverse led).
Naturally, the flipside of this flexibility are intensities that grow out of
bounds or become negative, which must be accounted for in the final translation
from intensity to pixel values.

\subsection{Numerical approximation of segments}
\seclab{numapproxseg}

In the preceding sections, we have established the mathematical tools required
to formulate our overall visualization procedure in \secref{visalgo}.
All that is required for an application is to produce values $(A, B)$ given a
ray segment running from $x_1$ to $x_2$.
We use this section to propose one way that is close to the true physics by
prompting the application for values of $\beta (x)$ and $\gamma (x)$ at
locations that are automatically decided by a user-provided accuracy
requirement.

There are two mathematical approaches to compute the change of intensity caused
by the medium along a ray segment.
One is numerically solving the ODE \eqnref{Iode}, and the other is to
integrating it analytically and evaluating the integral by quadrature.
We discuss both here, beginning with the ODE solve.

Considering the reformulation \eqnref{knownIAB}, we find two independent ODEs
\begin{equation}
  \eqnlab{ODEsAB}
  A' (x) = - \beta (x) A (x)
  , \qquad
  B' (x) = \gamma (x) - \beta (x) B (x)
  .
\end{equation}
%
Choosing an explicit Runge-Kutta (RK) method with $M$ stages and
coefficients $c_i$, $a_{ij}$, and $b_j$,
we obtain the evaluation points
\begin{equation}
  \xi_i = x_1 + c_i \Delta x , \qquad i = 0, \ldots, M - 1
\end{equation}
and the stage values
\begin{equation}
  \eqnlab{RKresultAB}
  \bar A_i  = 1 + \Delta x
    \sum_{j = 0}^{i - 1} a_{ij} \left( - \beta ( \xi_j) \bar A_j \right)
  , \qquad
  \bar B_i  = \Delta x
    \sum_{j = 0}^{i - 1} a_{ij}
    \left( \gamma ( \xi_j) - \beta (\xi_j) \bar B_j \right)
  .
\end{equation}
When selecting subdiagonal-explicit RK methods, such as explicit Euler, Heun's
method of order 2 or 3, or classical RK4, the sums collapse to one entry
containing $a_{i,i-1}$.
The equations for the final values also have the form above when setting
$a_{Mj} = b_j$, $j = 0, \ldots, M - 1$, and $a_{MM} = 0$, which leads to $A =
\bar A_M$, $B = \bar B_M$.

Explicit RK methods have a restriction on the interval $\Delta x$ to ensure
stability.
For the type of ODE \eqnref{ODEsAB}, the limit is of the order $1 / \beta$.
While proceeding with the solve, we watch the evaluations for a
condition $\Delta x \beta (\xi_j) > c_\rRK$.
If it arises, we abandon this RK step and replace it with two RK steps, each
with one half of the original length.
We repeat this procedure recursively for a problem-adaptive solve, knowing that
the recursion is guaranteed to terminate.
The computation is numerically stable if, conservatively, $c_\rRK \le 1/2$.
We may reduce $c_\rRK$ further to improve the accuracy of the method with an
error rate of $c_\rRK ^M$.
The adaptive RK solve is shown in \algref{adaptiverk}.
%
\begin{algorithm}[tbp]
  \KwIn{$\Delta x = x_2 - x_1 \ge 0$, stability parameter $c_\rRK > 0$}
  \For{$i = 0, \ldots, M$}{
    Compute RK stage $(\bar A_i, \bar B_i)$
    by \eqnref{RKresultAB}\;
    \If{$i > 0$ and above computation evaluates
        $\Delta x \beta(\xi_{i - 1}) > c_\rRK$}{
      $x_{1/2} = \frac12 (x_1 + x_2)$
        \tcc*{split at midpoint}
      \treturn $\tsegeval (x_{1/2}, x_2) \gplus
                \tsegeval (x_1, x_{1/2})$%
        \tcc*{recursion}
    }
  }
  \treturn $(\bar A_M, \bar B_M)$
  \caption{$(A, B) \leftarrow \tsegeval (x_1, x_2)$
           by adaptive ODE solve}
  \alglab{adaptiverk}
\end{algorithm}

To avoid the stability issue altogether, we may switch to an implicit RK
method.
For the Gau\ss{} method of $M = 2$ stages and order 4, we find the evaluation
points $\xi_{0/1} = \frac12 \pm \sqrt{3} / 6$ and solve the RK system
analytically for
\begin{equation}
  \eqnlab{gaussAB}
  A = \frac{1 - a_1 + a_2}{1 + a_1 + a_2} ,
  \qquad
  B = \frac{b_1 + b_2}{1 + a_1 + a_2} ,
\end{equation}
where
\begin{subequations}
\begin{align}
  \eqnlab{gaussABab}
  a_1 & = \Delta x \left( \beta (\xi_0) + \beta (\xi_1) \right) / 4 ,
  &
  a_2 & = \Delta x^2 \beta (\xi_0) \beta (\xi_1) / 12 ,
  \\
  b_1 & = \Delta x \left( \gamma (\xi_0) + \gamma (\xi_1) \right) / 2 ,
  &
  b_2 & = \Delta x^2
  \left(
    \beta(\xi_0) \gamma (\xi_1) - \beta (\xi_1) \gamma (\xi_0)
  \right) \sqrt{3} / 12 .
%
%
\end{align}
\end{subequations}
This formula is bounded regardless of $\Delta x$, but we see that the
limit of very large $\Delta x$ produces $A = 1$ and a difference of intensities
for $B$, both unphysical.
Thus, we use the recursion in \algref{adaptiverk} as before, which allows us to
enforce accuracy with rate $c_\rRK^{2M}$.

The approach by quadrature can be exemplified by using Simpson's rule,
characterized by $M = 3$ and known locations $c_i$ and weights $b_j$, as a good
compromise between simplicity, accuracy, and cost.
Considering the closed form \eqnref{ABgeneral}, we arrive at
\begin{equation}
  \eqnlab{simpson}
  A = \exp \Biggl( - \Delta x \sum_{j = 0}^2 b_j \beta (\xi_j) \Biggr) ,
  \qquad
  B = \Delta x \left( b_0 A \gamma (\xi_0) +
                      b_1 \sqrt A \gamma (\xi_1) +
                      b_2 \gamma (\xi_2) \right)
  .
%
%
\end{equation}
A disadvantage is that the computation involves exponentials and is thus more
expensive than the numerical ODE solve.
An advantage is that $\Delta x$ may be arbitrarily large and the result stays
physical.
We recommended again to embed the quadrature formula into the recursion
described above, since the accuracy control by $c_\rRK^M$ remains valuable.

\subsection{Visualization of surfaces}
\seclab{surfaces}

Surfaces can be understood as 2D manifolds idealized from very thin 3D objects.
Thus, a surface has a single intersection point with a ray, unless the ray is
parallel or the surface doubles back to meet the ray several times.
To model the optical parameters of the surface, we assume a surface thickness
$h \ll \Delta x$.
The distance traveled inside the surface is $h / \cos (\phi)$,
where $\phi$ is the angle between the ray and the surface normal at the
intersection point.
Since $h$ is so small, it is justified to assume constant coefficients $\beta$,
$\gamma$ for the surface's material.

If we prescribe a desired transparency factor $A_s$ and absolute intensity
$I_s$ produced by the surface, we obtain the value $B_s$ by
\eqnref{BIrelation}.
This computation is independent of the actual thickness assumed.
In other words,
we are not required to choose any particular value for
$h$, and the method is well defined for $h \to 0$.
%
This approach extends to rendering opaque surfaces.
Instead of trying to approximate the limit $\beta \to \infty$, we may directly
set $A_s = 0$ and $B_s = I_s$ with no change to the aggregation logic.
To incorporate the angle, we rely once more on \eqnref{ABconstant} and define
\begin{equation}
  \eqnlab{ABphi}
  A_\phi = A_s^{1 / \cos \phi} , \qquad B_\phi = B (A_\phi, I_s) .
\end{equation}

\section{Visualization Algorithm}
\seclab{visalgo}

In this section we go through the different phases of the overall parallel
visualization algorithm.
As mentioned above, we assume that the scene is stored in distributed memory,
where the partition is defined by that of the space-forest.
The relevant scene data is available from the individual leaves.
Each leaf is stored on precisely one process, on which it is local.
It is remote to all other processes.
These assumptions are naturally met when visualizing computational data
in-situ, and can be established by parallel forest traversal algorithms in the
case of discrete data such as point sets or constructive solid geometry (CSG)
descriptions.
We demonstrate one example for each type of data in \secref{examples}.

We consider the MPI model of parallelization, identifying an MPI rank with a
process.
The algorithms we present still apply if we further partition each MPI rank
between CPU cores and/or accelerator units, and use threads instead of
processes.
We generally expect that all subalgorithms run simultaneously on all processes.
This includes the case when some processes may have no relevant data to process
in one or more of the phases despite our efforts to load-balance.

We allow for multiple images created in one rendering, where each image
instance may have a different camera definition and/or different rules for
background/skybox data and attenuation and emission.
Top-down traversals of the forest will continue until the last of the instances
considers a leaf or subtree invisible.
A node is invisible when its AABB falls out of the image area, or when it is so
small that it is no longer hit by a single ray.

Due to the aggregation logic laid out in \secref{vismodel}, we do not need to
visit the elements in back-to-front ray order.
Instead, we follow the space filling curve encoding the distributed mesh,
which ensures linear and cache-friendly memory access.
We assume that the total number of MPI processes is given and fixed, and that
the user has specified another and usually smaller number of writer processes
to store the images on disk using parallel I/O.


The main challenge in parallel is thus to transport the image information from
the data-oriented partition of the input forest to the pixel-oriented partition
of the output image.
To this end, we propose the phases of reassigning the input data, rendering it
into ray segments, aggregating the segments, and compositing them to pixels.
Except for the actual rendering of segments, all phases require careful
redistribution of data in parallel
as described in the following.

\subsection{Culling and pre-partition}
\seclab{prepart}

The input forest is usually partitioned to equidistribute numerical data
between the processes.
Raycasting, on the other hand, will benefit from a parallel equidistribution of
visible pixels.
Since the most expensive operation in our visualization pipeline is the
numerical integration of ray segments for each element (\secref{numapproxseg}),
it makes sense to reassign the data in parallel before rendering.

%

We render non-destructively, meaning that we treat the input as read-only.
This first phase is to create a visualization forest (``vforest'') structure
repartitioned by pixel count, copying only the data of the visible elements
for the upcoming rendering step.
\begin{enumerate}
  \item Initialize an empty vforest, e.g.\ using
        \pforestbuild \cite[Section 3]{Burstedde18}.
  \item Run a top-down traversal of the local input elements, e.g.\ using
        \pforestlocal \cite[Section 3]{IsaacBursteddeWilcoxEtAl15},
        tracking which of the image instances consider a tree node visible.
  \begin{itemize}
    \item Each tree node has multiple AABBs, one for each image instance.
          Whenever the AABB of a node falls outside of an image instance's view
          or shrinks to zero pixels, that instance is removed from the list.
    \item When the list is empty, that is, the AABB of the tree node has become
          invisible in all instances, the node is pruned from the recursion.
    \item Whenever we reach a leaf, it is visible to at least one instance and
          we add it to the vforest and copy its element data.
  \end{itemize}
  \item Finalize the vforest by adding the smallest possible set of coarse
        invisible elements that make it a complete octree; this is still part
        of \pforestbuild.
  \item Repartition the vforest, where each leaf has a weight proportional
        to its visible pixel count summed over all instances.
  \item Communicate the data copied above to the new owner of its element,
        e.g.\ by calling the \pforesttrans functions
        \cite[Section 6]{Burstedde18}.
\end{enumerate}

The parallel forest build, repartitioning, and transfer of data are fast
operations for an efficient forest-of-octree implementation, requiring only a
fraction of a second.
The run time of copying the visible element data will be proportional to the
input set and generally a lot less costly than computing the data in the first
place.
Still, when in doubt, the entire procedure may be skipped since it is
technically optional.

We reason against skipping, however, by pointing out several benefits:
Depending on the camera positions relative to the domain, a relatively small
number of elements may be visible at all.
When using extreme adaptive mesh refinement, some parts of the domain may be
refined so deeply that many elements are not touched by any ray.
In both cases, the vforest may have a lot less leaves than the input forest,
considerably speeding up any subsequent traversal and computation.
The repartitioning by pixel count will help to improve the load balance of the
rendering phase described next for an additional speedup.

\subsection{Leaf-level rendering}
\seclab{leaflevel}


To turn simulation data into color ray information along a ray, we
process the local leaves of the vforest, again by using \pforestlocal.
It is communication free as it descends the local portion of each tree
top-down.
As in the previous \secref{prepart}, the purpose of the recursion is to cull
subtrees from the rendering process as early as possible by AABB checks of the
nodes.
Eventually we arrive at all local elements that are at least partially visible.
We derive the rectangle of possibly intersecting rays from an element's AABB
and use this information to compute the color information tuples $(A, B)$ by
\secref{numapproxseg} and \secref{surfaces}, one leaf at a time.
Each is thus associated with a rectangle of rays, where each ray may be
characterized by zero, one, or more segments.
If there are more than one, which happens rarely and only for non-linear
geometries, we sort this small number by $z$ coordinate; this is the only
explicit sort we do.

This rendering phase is dominated by computing the ray intersections of the
local leaves contained in their respective bounding boxes.
The computational work is linear in the number of local elements weighted by
the number of rays per element.
%
This may no longer be the case if we use an adaptive integration to
achieve a prescribed accuracy $\beta \Delta t < c$ as proposed in
\algref{adaptiverk}.
Then elements with a size larger than a certain threshold will integrate more
than one intermediate segments in the recursion before returning the resulting
segment.
We can correct for this effect during the weighted partitioning described in
\secref{prepart}.

%

\subsection{Coarsening and post-partition}
\seclab{postpart}

At this point, we have augmented the visible elements with the relevant ray
segments and computed their color information in terms of the per-channel
coefficients $(A, B)$.
The number of visible elements is
unchanged from the space-forest that we received as input.  This number, even
when load-balanced by the local segment count, may be much too large globally
to proceed directly to the compositing step described in \secref{compositing}
below.

We observe the following, which applies as such in three dimensions only:
If we coarsen the vforest by replacing a family of same-size leaves with their
common parent, we reduce the number of elements by the factor eight.
This saves memory for all elements, visible or not.
The ray area of the parent becomes the union of the ray areas of all children.
By using the aggregation of all segments for a given ray by
\secref{cutoverlap}, contributed by all children that intersect this ray, we
compute the equivalent color information for the parent.
In doing so, we reduce the overall number of segments by roughly a factor two,
at a cost that is linear in the number of incoming child segments.
Due to our aggregation formulas, the process is lossless up to roundoff error.


After one aggregation step, we may wish to repartition the elements, again
weighted by their (updated) ray count.
If we do this, we to marshal/send and receive/unmarshal the segments to move
them to the new owners of the associated elements.
We can repeat this cycle of coarsening, aggregation, partition, and transfer
until the number of elements per process would become too low to be efficient,
or until the number of ray segments would become difficult to load-balance due
to the decreased element granularity.
In each stage, the communication volume and memory used by the ray segments is
roughly half that of the previous stage.
Generally, this aggregation algorithm is optional and configurable, for example
by restricting the maximum number of cycles.
In two dimensions, we deal with point intersections instead of segments, but
the algorithm may still be beneficial in case of many invisible elements.


\subsection{Compositing}
\seclab{compositing}

We face the challenge to complete the aggregation of ray segments until we
obtain the final pixel color for each ray.
While we have condensed the overall number of visible elements and ray segments
by coarsening, most rays still have their remaining segments scattered between
multiple processes according to the partition of the vforest.
Furthermore, the parallel partition of elements depends on their positions in
space and is not aligned geometrically with the camera or the pixels of the
image plane.

To resolve this situation, we introduce an independent MPI partition that
divides the images into a configurable number of equal-size tiles, each of
which we assign to a distinct writer process.
The number of writer processes can be chosen to optimize the write time using
MPI I/O over a subcommunicator.
Too few writers will not deliver the maximum I/O bandwidth, while too many will
suffer increasing network load.

The image partition is configured by the dimensions of the image instances and
the pixel count of the tiles and thus known to all vforest processes at no
extra cost.
In order for these to send their ray segments to the matching writer processes
in the image partition, we have to determine the pairs of sending and receiving
processes.
We aim to compute the minimal set of such pairs, such that each process sends
segments to the writers responsible for only the relevant tiles.
The set of matching writers is thus different for every vforest process and not
known a priori.
In the following, we describe how we determine all pairs without any
communication, which keeps the algorithm synchronization-free and scalable.

\subsubsection{Senders determine the receiving processes}
\seclab{sending}

Every process $p$ in the vforest partition must determine the writer processes
$q'$
whose tiles overlap with any of $p$'s elements' bounding boxes.
This is necessary so the ray segments can be marshalled and sent as MPI
messages.
The task can be executed in a top-down and process-local manner similar to the
rendering phase described in \secref{leaflevel} except for the stopping
criteria.
In fact, should we be skipping the coarsening and post-partition of
\secref{postpart}, we would piggy-back this logic onto the rendering recursion.

In \pforestlocal we compute the AABB for every node of the tree that we
enter, starting with its root.  Now, if the AABB does not overlap the image at
all, we stop the recursion at this node.  If the AABB is fully contained in one
of the image tiles, we record its owner process as a receiver and stop the
recursion as well.  Now, if the recursion is still active, we know that more
than one tile overlap the node's bounding box.  If this node is a leaf, we
record the owner of each tile, marshal its segments accordingly and stop.
Else, we descend into the recursion for each child node.

We must arrange for the appropriate number of send buffers to hold
the marshalled ray segments.
After the recursion has completed, $p$ posts one non-blocking MPI send to
each receiver $q'$.

\subsubsection{Receivers determine the sending processes}
\seclab{receiving}

By symmetry, every writer process $q$ in the image partition must determine the
processes $p'$ in the vforest partition that hold ray segments overlapping
$q$'s tile.
This is necessary so the ray segments can be received as MPI messages into
a known number of pre-allocated slots.

In principle, this operation mirrors the one from the above \secref{sending}.
(In fact, it does not matter in which sequence the two are executed---they
might even be processed concurrently using a threaded MPI implementation.)
However,
from the perspective of a process responsible for an image tile, it is tricky
to collect the necessary information on the geometry/tile overlap, since many
or all elements relevant for a tile are stored on other processes.
Due to the distributed storage of elements and the segments associated, these
cannot be known to $q$.



What comes to the rescue is the metadata on partition markers that is
maintained in the \pforest library.
This metadata is available on every process and encodes the location of each
process' first leaf in the space-forest, more precisely, the number of the tree
and the lower left front coordinates of each process' first element.
Even though the leaves themselves are distributed and only known to their
respective owner processes, every writer process can execute a virtual top-down
recursion over the partition based solely on the metadata.
Moreover, since the geometry is defined per-tree and accessible to all
processes, we can apply the geometry transformation to any of the virtual
nodes encountered in the traversal
to compute their AABB.

Beginning with the root of each tree, we run a binary search over the markers
to identify the range of processes that own this tree's leaves.
We then
descend into the recursion, where we stop if the bounding box of the current
node does not intersect the writer's tile.
Otherwise we check if the range contains just a single process.
In this case, we add it to the list of senders and stop the recursion.
On the other hand, if the range is non-trivial, the node must be a subtree and
we proceed with the recursion into each child, where we split the range of
partition markers further using a nested binary search.
It is important to recall that the nodes we process here are purely virtual in
the sense that we have no access to their owner processes' data.
We just use their coordinates to compare their position to the partition markers.
By construction, this approach terminates and discovers all processes that will
be sending at least one ray segment to the writer process.
Also by construction, each writer process will exclusively receive segments
inside its own tile.

The complete logic of this partition traversal is detailed elsewhere
\cite[Section 4]{Burstedde18}.
After the sending processes are determined, each writer process allocates the
appropriate number of receive buffers and posts one non-blocking MPI receive
for each sender.

\subsubsection{Completion}
\seclab{completing}

The algorithms described above allow us to communicate the ray segments in a
decentralized point-to-point pattern between known senders and receivers.
Since we use non-blocking MPI, a writer process can overlap the unmarshalling
and pixel-wise aggregation of ray segments that have already been received with
the time waiting for the arrival of the remaining segments.
After this is done, each writes its tile of the image to the output file.
Here we use the uncompressed netpbm format \cite{MurrayvanRyper96} that we can
address with the standard MPI I/O mechanism based on offsets and strides over
one large image file per instance.


We summarize the parallel visualization pipeline in \algref{pipeline}.
For clarity, we do not detail the generalization to multiple image instances,
and we omit the MPI wait and I/O open/close calls.
Since we are using non-blocking MPI, the sending and receiving parts may be
swapped or run simultaneously using threads.
Strictly speaking, culling/pre-partition is optional as well as
coarsen/post-partition.
If we omit the latter, we can coalesce the respective top-down recursions of
\secref{leaflevel} and \secref{sending}.
\begin{algorithm}[tbp]
\tcc*{cull + pre-partition (\secref{prepart})}
  Top-down recursion over local nodes
  \Begin {
    stop (if AABB is not intersecting image)\;
    leaf (add element to vforest, copy data)\;
  }
  Complete vforest; work with it from now on\;
  Partition (elements weighted by number of ray segments)\;
  Transfer (data copied to new owners)\;
\tcc*{leaf-level rendering (\secref{leaflevel})}
  Top-down recursion over local nodes
  \Begin {
    stop (if AABB is not intersecting image)\;
    leaf (compute ray segments for element)\;
  }
\tcc*{coarsen + post-partition (\secref{postpart})}
  \For {$[0, \ldots, \text{\# coarsening iterations})$ or until too few elements}{
    Coarsen (by one level; aggregate child segments to parent)\;
    Partition (elements weighted by number of ray segments)\;
    Transfer (ray segments to new owners)\;
  }
\tcc*{all processes send (\secref{sending})}
  Top-down recursion over local nodes
  \Begin {
    stop (if AABB is not intersecting image)\;
    stop (if AABB contained in single tile owned by $q'$: record as receiver)\;
    leaf (record owners $q_i'$ of all tiles that AABB intersects%
          )\;
  }
  Post non-blocking sends of ray segments\;
\tcc*{writer processes receive (\secref{receiving})}
  \If {this is a writer process responsible for an image tile} {
    Top-down recursion over virtual nodes
    \Begin {
      stop (if AABB of node is not intersecting the tile)\;
      stop (if range of owners of node is a single process $p'$:
            record as sender)\;
    }
    Post non-blocking receives of ray segments\;
    Receive segments and aggregate into tile's pixels\;
    Use parallel I/O to write tile to image file\;
  }
  \caption{Parallel visualization}
  \alglab{pipeline}
\end{algorithm}

\subsection{A dynamic tree data structure for rectangles of ray segments}
\seclab{tiles}

We have so far not provided any details on the data structure we use to store
rectangles of ray segments.
This data structure must have the following features:
\begin{enumerate}
  \item
    It begins life when the rays are cast through the pixels of the element's
    AABB.
    Some rays in the box will not hit the element and produce empty pixels,
    while others produce the $(A, B)$ values for one segment.
    When using non-linear geometries, one element may occasionally give rise to
    two or more segments for the same ray;
    when compositing, we will be combining multiple ray segments for the same
    pixel.
    Thus, the data structure must be able to store a variable number of
    segments per pixel.
  \item
    When we composite the ray segments of two elements that have overlapping
    views, we need to be able to create the union of the two data structures.
    Even if each one is a rectangle, the union may not be.
    Since many more than two elements may overlap, the union will get
    progressively more ragged in shape.
    Thus, we must be able to encode and store the union of arbitrarily many
    rectangles.
  \item
    When creating the union of two structures, each pixel in the intersection
    must be processed to aggregate the ray segments from the two sources by the
    procedure displayed in \secref{cutoverlap}, reducing their number as much as
    possible.
  \item
    Most elements are small in relation to the image.
    Storing the rectangle of all image pixels where only a small fraction
    contains any segments would be wasteful and impractical.
    Some compression of empty pixels will be useful.
  \item
    When communicating the segments for final compositing and writing, segments
    covered by different writer tiles are sent to different processes.
    It will be advantageous if the data structure respects the boundaries
    between the tiles for easier marshalling.
\end{enumerate}

One way of satisfying the above requirements is to create an adaptive quadtree
where a node at level $\ell$ corresponds to a rectangle of edge length $2^{L -
\ell}$.
Here $L$ is chosen smallest such that the image is contained in the unit square
of length $2^L$.
A node at level $\ell = L$ corresponds to a single pixel.
When the image is rectangular or not an exact power of two in length, pixel
nodes may be inside or outside the image, while nodes with $\ell < L$ may be
entirely outside the image, entirely inside, or split.
When rendering several image instances simultaneoeusly, each is addressed by a
tree of the matching maximum depth $L$.

Each tree leaf may have a payload that is a rectangular subset of its image
area.
When inserting a new rectangle of ray segments into the tree, we begin a
recursion at the root and stop when we reach a leaf without payload that fits
around the rectangle, or a leaf where the union of its payload and the
rectangle is again rectangular and fits into the node's area.
Otherwise, we may have a leaf which we refine into its four children.
We split the rectangle to insert into four according to the children's image
areas and insert each of them by recursion.
We realize the union of two trees by a similar recursive procedure.

For each of the pixels of a rectangular payload, we store a variable number of
ray segments ordered by $z$ coordinate.
When merging a pixel common between existing and incoming payloads,
we proceed along their two sets of segments in a common loop akin to one step
of merge sort, which has a run time that is linear in the sum of existing and
incoming segments.

To align the tree subdivision with the image tiles to write, we choose all tile
lengths $2^t$ as the same power of 2.
Thus, they align with nodes of the image tree at depth $L - t$.
Since the number of writer processes is prescribed and must bound the number of
tiles from above, we take the smallest exponent $t$ that satisfies this.
Here we only count the tiles that overlap the image at least partially.
For multiple image instances, the number of visible tiles per instance
varies and we require their sum to be less equal the number of writers.

This data structure satisfies all of our requirements.
In particular, the data size when marshalling a tree is proportional to the
number of ray segments contained and does not depend on the dimension of the
image.
We implement the trees dynamically using caching where we can since millions or
even billions of these structures are created, split, merged, and destroyed in
one rendering.

%

\section{Numerical Examples}
\seclab{examples}

In this section we construct two different examples of data that we render
using an implementation of the above ideas.
We use sythetic data that is governed by a distributed adaptive space-tree
whose depth and detail we can configure freely.
Our focus is on scaling up the method, increasing the amount of data, the
number of MPI processes holding the data, and the number of writer processes
writing the image files.
We verify that the run times and the number of ray segments processed stay
bounded as expected.

\subsection{The Mandelbrot set}
\seclab{fractal}

Our first example renders a 2D surface positioned in 3D space.
We choose a fractal data set to provoke a rather irregular distribution of
elements, namely the well-known Mandelbrot set.
It is a subset of the complex plane, and a number $c \in \sC$ belongs to the
set if the iteration
\begin{equation}
  \eqnlab{mbiter}
  z_{n + 1} = z_n^2 + c
\end{equation}
stays bounded when applied to the initial value $z_0 = 0$.
The Mandelbrot set is thus contained in a circle of radius 2 around the origin
and symmetric to the real axis.

We define a space-forest of two quadtrees and a geometry transformation that
maps them to a hexagon in the $xy$ plane that we consider a part of $\sC$.
Its top corners are
\begin{equation}
  (-23/10, 3/4, 0), \quad (-1/4, 9/5, 0), \quad (13/10, 1, 0) ,
\end{equation}
and its bottom corners are derived by negating the $y$ coordinate.

We begin with a uniform refinement at a given level.
The mesh is augmented with piecewise bilinear finite elements whose degrees of
freedom are the non-hanging vertices in the mesh.
We use $xy$ coordinates of each vertex as $c$ and execute the iteration
\eqnref{mbiter} until it exits the 2-circle or until a maximum number of
iterations is reached.
This number is then divided by the maximum to yield a real number in $[0, 1]$.
One element has four such numbers at its corners that we use for bilinear
interpolation when we intersect the element with a ray passing through.
The result is converted to $(A, B)$ values separately for the three RGB
channels.

Now we run multiple cycles of refinement, repartitioning, and recomputation to
increase the resolution of the data.
Specifically, we refine every element whose four corner values are not all
equal and re-establish the 2:1 balance of neighbor element sizes.
The Mandelbrot iteration is performed anew and we run the next raycasting.
Throughout these adaptation cycles, we configure two image instances of the
same width and height, differing in camera position/angle and transfer
function.
We display some results in \figref{mbsetpics}.
\begin{figure}
    \includegraphics[width=67mm]{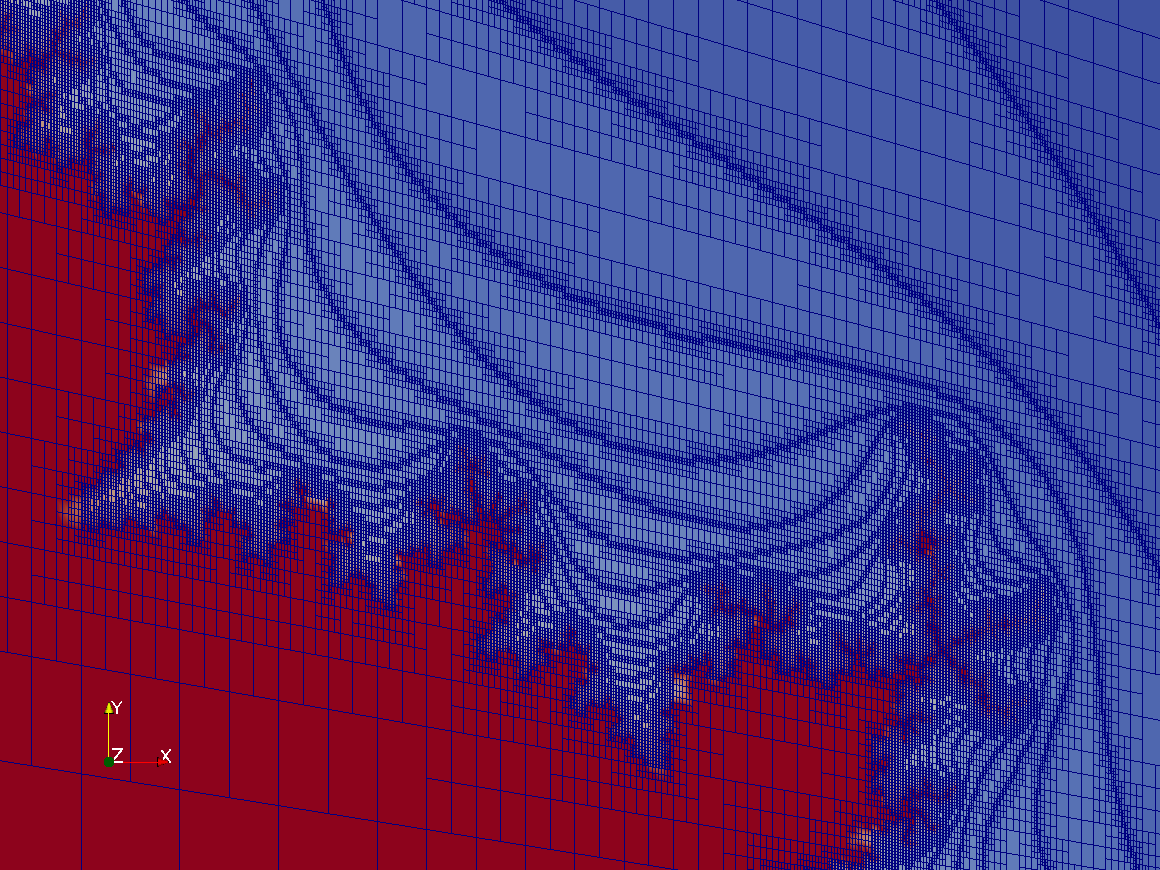}
    \hfill
    \includegraphics[width=67mm]{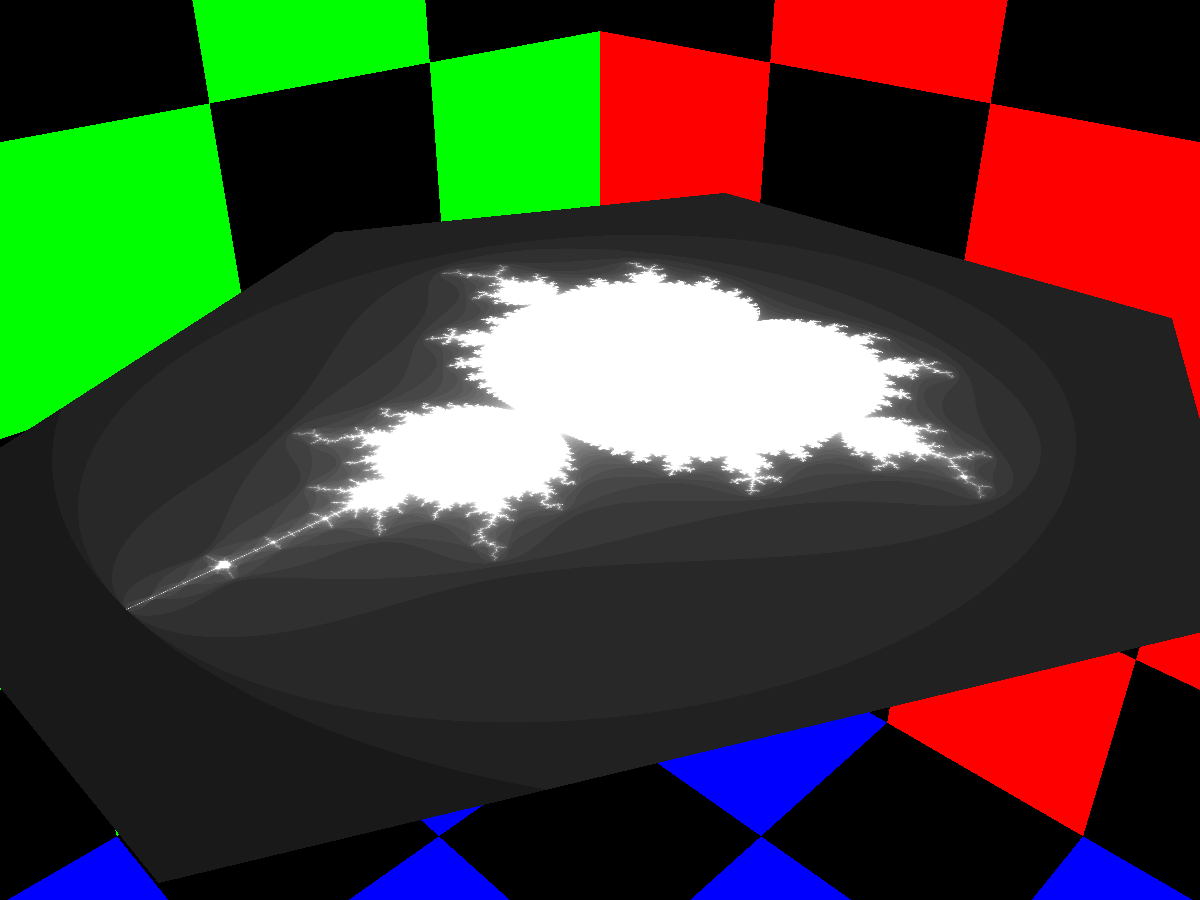}
    \\[2ex]
    \includegraphics[width=67mm]{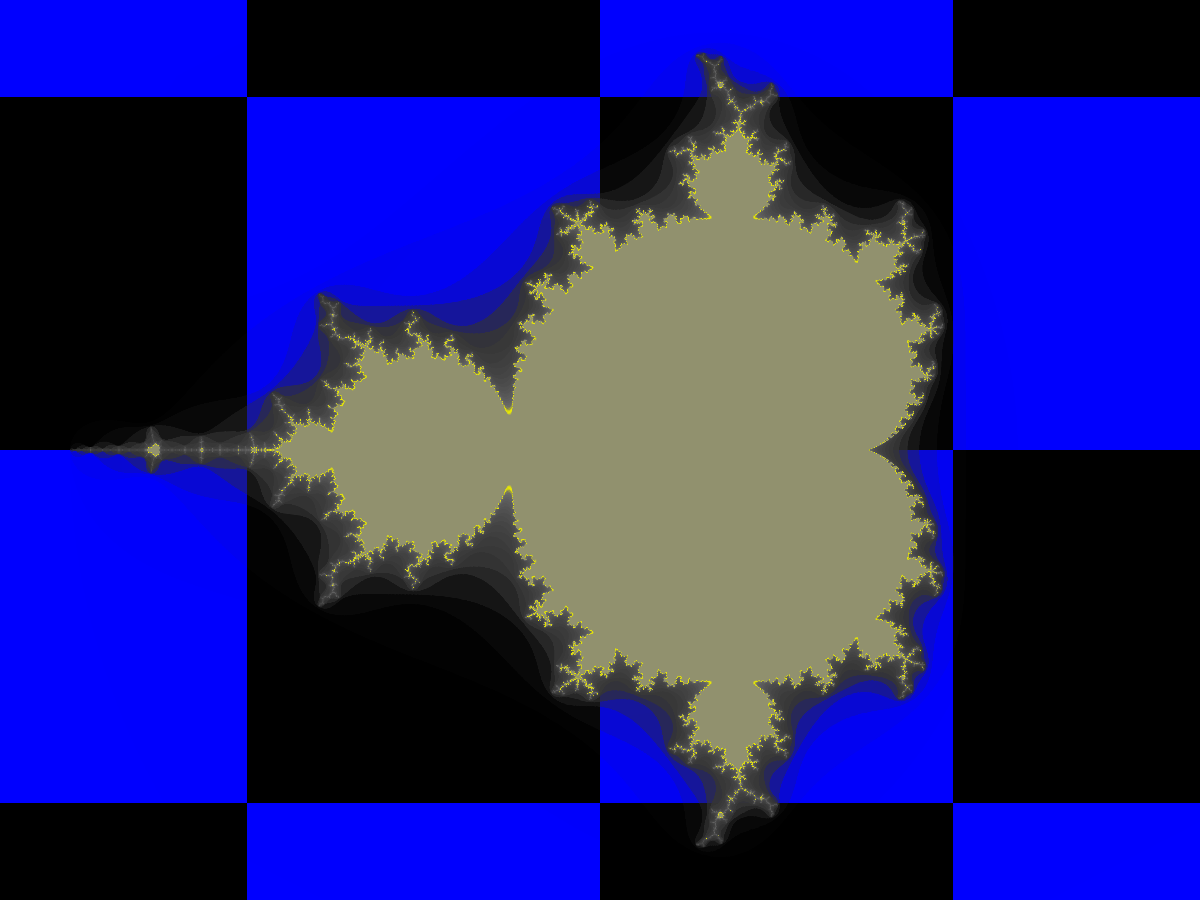}
    \hfill
    \parbox{67mm}{\mbox{}}
  \caption{Rendering of the Mandelbrot set with an iteration limit of 30.
    The top left shows a standard visualization based on writing VTK files,
    showing a zoom into the adaptive 2D mesh refined to level 11.
    Locations drawn red have reached the iteration limit and are considered
    inside the set.
    Blue locations have values less than 1, and we see that the decrease in
    steps of 1/30 is reflected in the refinement pattern.
    Top right: One image instance computed by raycasting, using an RGB
    checkerboard as skybox.
    The manifold is rendered opaque with a color scale from grey to white.
    Bottom left: For the other instance we use transparency and a color model
    that emphasizes an iteration count just below the maximum in yellow.%
    }%
  \figlab{mbsetpics}%
\end{figure}%

%
%
%
%

\subsection{Randomly distributed spheres}
\seclab{spheres}

Our 3D example is set in a cube that we populate with spheres of random
position and diameter.
Again, we aim at a scale-invariant design that we realize by prescribing the
probability density $\rho (s)$ of the sphere's cross section $s$.
When the spheres get smaller, we want to have more of them to give them the
same optical density, which lets us choose $\rho (s) = 1/s$.
If we choose a minimum $s_0$ and maximum $s_1$, the expected cross section $E$
is available in closed form.

We begin with a uniformly refined mesh as before and let each process iterate
through the local elements.
Dividing the cross section of the element $C$ by $E$ gives us the number of spheres
we would expect in this element.
Running a Poisson sampler with mean $\lambda = C/E$ gives us the number of
spheres to create in the element.
For each of the spheres to create, we sample the radius according to $\rho$ and
choose its center coordinates randomly inside the element.
We implemented the samplers based on excellent material
\cite{PressTeukolskyVetterlingEtAl99}.---
We allow for spheres that extend out of the element or out of the unit cube,
and spheres that overlap with other spheres.

We focus the boundary of the spheres using a fairly sharp transfer function.
The adaptation cycles proceed by refining any element that intersects any
sphere's shell between the radii $[1 - \varepsilon, 1 + \varepsilon] r$,
except when the element is already smaller than $\varepsilon r / 2$.
Since a shell generally touches multiple elements on remote processes,
the partition traversal \pforestparti proves valuable again.
We skip the details of repartitioning and communicating the sphere data
to all relevant elements.
Essentially, we ensure that any process only stores spheres whose shell
intersects its elements, and that the number of element/sphere checks stays
optimal.
In this example, the mesh is not 2:1 balanced.

We illustrate the effect of adaptive meshing and adaptive integration
in \figref{spherelevels}.
If the refinement is not sufficiently deep to resolve the spheres, the ray
integration does not sample the shells densely enough.
If we refine more deeply and choose a smaller threshold $c_\rRK$, the
artefacts disappear.
\begin{figure}
    \mbox{}\hfill
    \includegraphics[width=52mm]{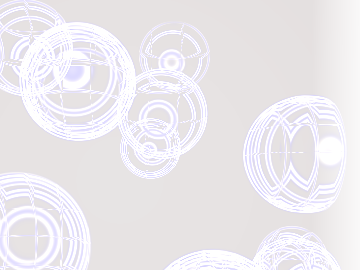}
    \hspace{3ex}
    \includegraphics[width=52mm]{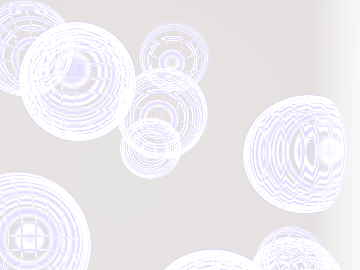}
    \hfill\mbox{}
    \\[2ex]
    \raisebox{4.5mm}{\includegraphics[height=58mm]{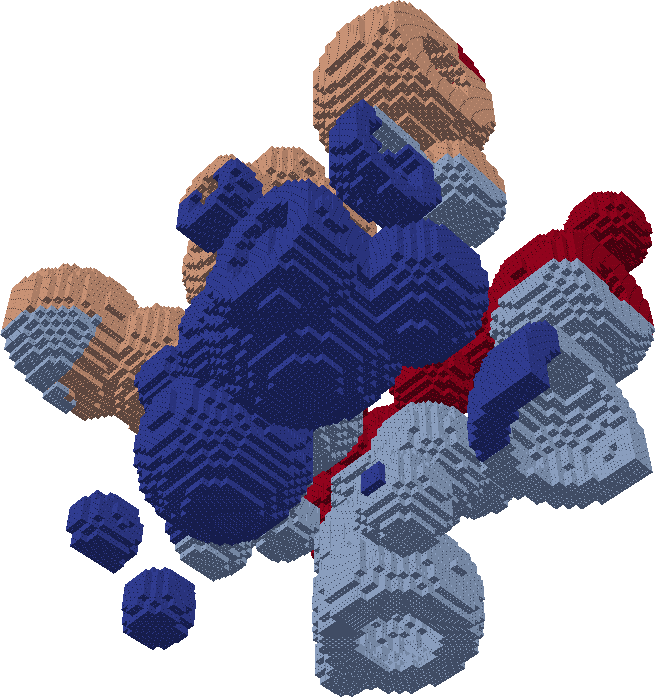}}%
    \hfill
    \includegraphics[height=67mm]{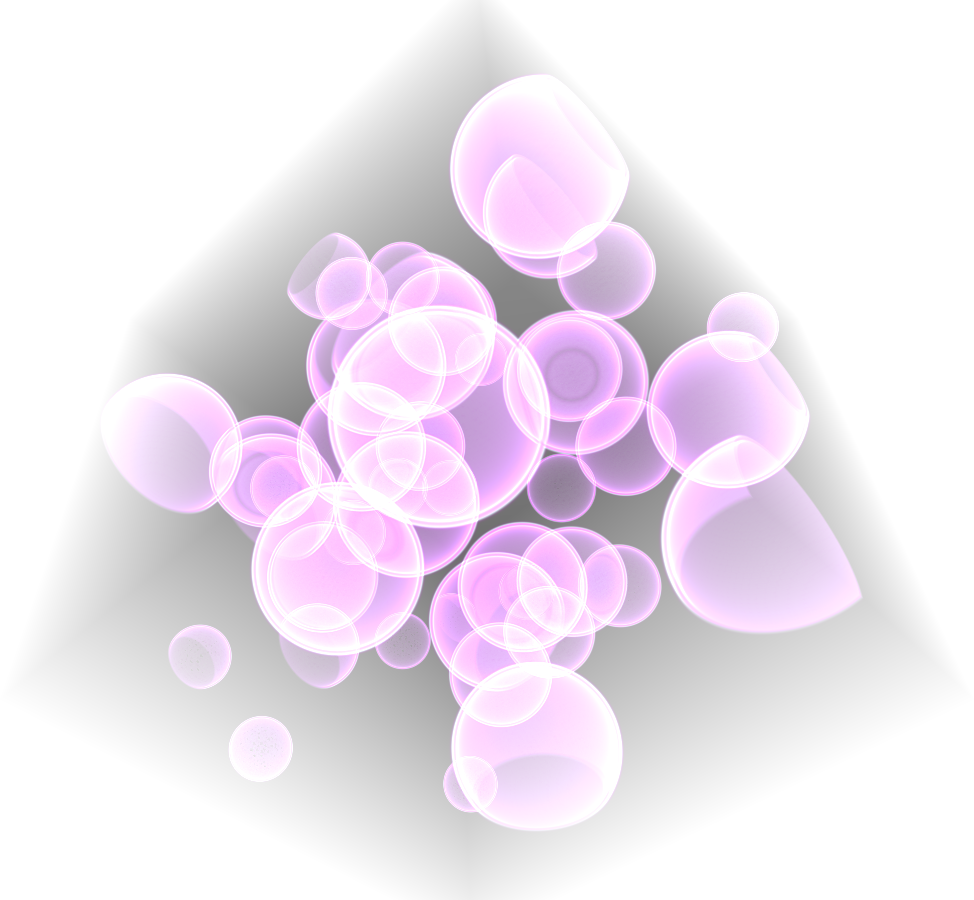}%
  \caption{%
    Top row: adaptive mesh to maximum level 4 (left) and 5 (right);
    rays integrated with Simpson's rule and threshold $c_\rRK = 1$.
    We observe aliasing due to insufficient refinement and undersampling.
    Bottom row: setup with 44 spheres refined to level 7.
    Left: VTK plot of the level 7 elements partitioned on 4 processes
    (color coded).
    Right: high-accuracy rendering using the implicit Gau\ss{} method
    and $c_\rRK = .02$.%
  }%
  \figlab{spherelevels}%
\end{figure}%

\section{Conclusion}
\seclab{conclusion}

%
%
%
%


\begin{acks}
The author would like to thank
Lucas C.\ Wilcox for suggesting to look at Chebfun's bivariate root finding
algorithms and A.\ Kraut for her feedback on Poisson distributions.
We would also like to thank Tobin Isaac for
co-maintaining and extending the \pforest software.

%
All \pforest algorithms used in this paper are available from its public
software repository \cite{Burstedde18a}.
We will make the building blocks of the visualization code, in particular the
adaptive tree structure for storing rectangles of ray segments, available after
modularizing it out of our example code.  This may, however, take some time.
\end{acks}

\bibliographystyle{ACM-Reference-Format-Journals}
\bibliography{group,ccgo_new}


\begin{thebibliography}{00}


\ifx \showCODEN    \undefined \def \showCODEN     #1{\unskip}     \fi
\ifx \showDOI      \undefined \def \showDOI       #1{{\tt DOI:}\penalty0{#1}\ }
  \fi
\ifx \showISBNx    \undefined \def \showISBNx     #1{\unskip}     \fi
\ifx \showISBNxiii \undefined \def \showISBNxiii  #1{\unskip}     \fi
\ifx \showISSN     \undefined \def \showISSN      #1{\unskip}     \fi
\ifx \showLCCN     \undefined \def \showLCCN      #1{\unskip}     \fi
\ifx \shownote     \undefined \def \shownote      #1{#1}          \fi
\ifx \showarticletitle \undefined \def \showarticletitle #1{#1}   \fi
\ifx \showURL      \undefined \def \showURL       #1{#1}          \fi

\bibitem[\protect\citeauthoryear{Burstedde}{Burstedde}{2010}]%
        {Burstedde18a}
{Carsten Burstedde}. 2010.
\newblock {\texttt{p4est}}: Parallel {AMR} on Forests of Octrees.
\newblock   (2010).
\newblock
\newblock
\shownote{\url{http://www.p4est.org/} (last accessed September 3rd, 2018).}


\bibitem[\protect\citeauthoryear{Burstedde}{Burstedde}{2018}]%
        {Burstedde18}
{Carsten Burstedde}. 2018.
\newblock Parallel tree algorithms for {AMR} and non-standard data access.
\newblock   (2018).
\newblock
\newblock
\shownote{http://arxiv.org/abs/1803.08432.}


\bibitem[\protect\citeauthoryear{Burstedde and Holke}{Burstedde and
  Holke}{2017}]%
        {BursteddeHolke17}
{Carsten Burstedde} {and} {Johannes Holke}. 2017.
\newblock \showarticletitle{Coarse mesh partitioning for tree-based {AMR}}.
\newblock {\em SIAM Journal on Scientific Computing\/} {39}, 5 (2017),
  C364--C392.
\newblock
\showURL{%
\url{https://doi.org/10.1137/16M1103518}}


\bibitem[\protect\citeauthoryear{Isaac, Burstedde, Wilcox, and Ghattas}{Isaac
  et~al\mbox{.}}{2015}]%
        {IsaacBursteddeWilcoxEtAl15}
{Tobin Isaac}, {Carsten Burstedde}, {Lucas~C. Wilcox}, {and} {Omar Ghattas}.
  2015.
\newblock \showarticletitle{Recursive algorithms for distributed forests of
  octrees}.
\newblock {\em SIAM Journal on Scientific Computing\/} {37}, 5 (2015),
  C497--C531.
\newblock
\showDOI{%
\url{http://dx.doi.org/10.1137/140970963}}


\bibitem[\protect\citeauthoryear{Murray and vanRyper}{Murray and
  vanRyper}{1996}]%
        {MurrayvanRyper96}
{James~D. Murray} {and} {William vanRyper}. 1996.
\newblock {\em Encyclopedia of Graphics File Formats\/} (second ed.).
\newblock O'Reilly.
\newblock
\showISBNx{1-56592-161-5}


\bibitem[\protect\citeauthoryear{Press, Teukolsky, Vetterling, and
  Flannery}{Press et~al\mbox{.}}{1999}]%
        {PressTeukolskyVetterlingEtAl99}
{William~H. Press}, {Saul~A. Teukolsky}, {William~T. Vetterling}, {and}
  {Brian~P. Flannery}. 1999.
\newblock {\em Numerical Recipes in C\/} (second ed.).
\newblock Cambridge University Press.
\newblock


\end{thebibliography}


\appendix
\setcounter{section}{0}

\section{Mathematical details and procedures}
\seclab{appmath}

\subsection{Computing the intersection of a ray with a tensor-product surface}
\seclab{intersection}

For any ray vector $\vvec r$ we can construct two vectors to form the
orthogonal system $(\vvec r, \vvec d_1$, $\vvec d_2)$.
We use them to project the three coordinate equations for the intersection of
the ray with a parameterized surface, written in unknowns $(\alpha, t_1, t_2)$
as
\begin{equation}
  \eqnlab{intersectionthreed}
  \pvec o + \alpha \vvec r =
  \pvec J_k (t_1, t_2) =
  \sum_{m_1, m_2 = 0}^R \pvec Z_{k; m_1, m_2}
    \psi_{m_1} (t_1) \psi_{m_2} (t_2)
  ,
\end{equation}
into two equations in $(t_1, t_2)$,
\begin{equation}
  \eqnlab{intersectiontwod}
  0 =
  \sum_{m_1, m_2 = 0}^R P_{j; m_1, m_2}
    \psi_{m_1} (t_1) \psi_{m_2} (t_2)
  ,
  \quad j = 1, 2
  ,
\end{equation}
where $P_{j; m_1, m_2} = \vvec d_j \cdot \pvec Z_{k; m_1, m_2}$.
We can think of this as two polynomial equations in $t_1$ with
coefficients in $t_2$, of which we have to find all common zeros $t_{1,i}$.
Once we achieve this, we can insert each zero back into the system and
solve for common zeros in the second variable.

The detailed procedure is known in commutative algebra as the resultant method
\todo{cite this}.
Any common zero implies a common linear factor in the two polynomials, and the
conditions for the existence of such a factor can be formulated in terms of
a matrix determinant.
A bilinear face parameterization yields at most two intersection points,
which can be computed by solving quadratic equations.
A biquadratic parameterization yields at most eight intersections, which
requires numeric root finding by bisection and/or Newton's method, combined
with successive division by linear factors.

\subsection{Computing bounding boxes for tensor product polynomials}
\seclab{boundingboxes}

Given a bilinear or trilinear geometry parameterized by Lobatto points (using
$\xi_{j} = j$, $j = 0, 1$, and $N = 1$ in \eqnref{nodal}, \eqnref{tensorgeom}),
the transformation from the reference element into space is a convex
combination of the corner points.
Thus, the bounding box is computed by a component-wise minimization and
maximization of these points.

For higher order geometries we can reduce the three-dimensional case to
considering each face in turn, at the core requiring to compute the bounding
box of a two-dimensional parameterized surface.
Maxima and minima can either occur at the corner points, along each edge,
or inside the surface area.
Corners are taken into account as is, by the same reasoning as above.
Extrema along an edge can be identified by performing one-dimensional root
finding of the derivative with respect to the edge coordinate.
The extrema inside the surface can be found by taking the derivatives of the
parameterization and performing the bivariate root finding procedure described
in \secref{intersection}.
In fact, much of the mathematics and the code can be reused.

\received{Month 2018}{Month 2018}{Month 2018}

\end{document}